\documentclass[a4paper,fleqn]{llncs}

\title{VLDL Satisfiability and Model Checking via Tree Automata\thanks{Supported by the project ``TriCS'' (ZI 1516/1-1) of the German Research Foundation (DFG).}}
\author{Alexander Weinert}

\institute{Reactive Systems Group, Saarland University, 66123 Saarbr{\"u}cken, Germany\\
 \email{weinert@react.uni-saarland.de}}
 
\usepackage{amsmath}
\usepackage{amssymb}
\usepackage{mathtools}
\usepackage{xspace}
\usepackage[utf8]{inputenc}
\usepackage{booktabs}
\usepackage{enumitem}
\usepackage{fullpage}

\usepackage{tikz}
\usetikzlibrary{automata,calc,decorations.pathreplacing}
\tikzset{initial text={}}
\tikzset{every initial by arrow/.style={-stealth}}

\usepackage{url}

\newcommand{\myquot}[1]{``#1''}

\newcommand{\initmark}{I}

\newcommand{\tikztree}[3]{\begin{tikzpicture}[baseline,thick]
	\node (root) at (0,.5) {#1};
	\node (lhs) at (-.5,-.5) {#2};
	\node (rhs) at (.5,-.5) {#3};
	
	\draw 	(root) -- (lhs)
			(root) -- (rhs);
\end{tikzpicture}}

\newcommand{\StackTree}{\mathrm{st}}
\newcommand{\height}{\textsc{height}}
\newcommand{\replacedby}{\leftarrow}
\newcommand{\node}{\textsc{node}}

\newcommand{\compl}[1]{\mkern 1.5mu\overline{\mkern-1.5mu#1\mkern-1.5mu}\mkern 1.5mu}

\newcommand{\direction}{\rightsquigarrow}
\newcommand{\direct}{\rightarrow}
\newcommand{\jump}{\curvearrowright}

\newcommand{\coloneq}{\mathop{:=}\,}

\newcommand{\bigo}{{\mathcal O}}

\renewcommand{\epsilon}{\varepsilon}

\newcommand{\nats}{\mathbb{N}}
\newcommand{\card}[1]{{|#1|}}

\newcommand{\set}[1]{\{#1\}}
\newcommand{\restr}[2]{{\left.\kern-\nulldelimiterspace #1 \vphantom{\big|} \right|_{#2} }}
\newcommand{\range}{\mathit{range}}

\newcommand{\bools}{{\mathbb{B}}}

\newcommand{\aut}{\mathfrak{A}}
\newcommand{\treeaut}{\mathfrak{T}}

\newcommand{\vpdalphabet}{{\widetilde{\Sigma}}}
\newcommand{\calls}{\Sigma_c}
\newcommand{\returns}{\Sigma_r}
\newcommand{\locals}{\Sigma_l}

\newcommand{\vpa}{\text{VPA}\xspace}
\newcommand{\vpas}{\text{VPAs}\xspace}

\newcommand{\tnvpa}{\text{TVPA}\xspace}

\newcommand{\oneaja}{\text{1-AJA}\xspace}
\newcommand{\oneajas}{\text{1-AJAs}\xspace}

\newcommand{\acc}{\mathit{acc}}
\newcommand{\rej}{\mathit{rej}}

\newcommand{\graph}{G}

\newcommand{\lang}{{\mathcal L}}

\newcommand{\vpsys}{{\mathcal S}}

\newcommand{\halfthinspace}{{\kern .08333em}}

\newcommand{\cl}{\mathrm{cl}}

\newcommand{\ddiamond}[1]{\langle #1 \rangle}
\newcommand{\bbox}[1]{[\halfthinspace #1 \halfthinspace]}

\newcommand{\Raut}{\mathcal{R}}

\newcommand{\ltl}{\text{LTL}\xspace}

\newcommand{\vldl}{\text{VLDL}\xspace}
\newcommand{\vltl}{\text{VLTL}\xspace}

\newcommand{\exptime}{\textsc{{ExpTime}}\xspace}

\newcommand{\threeexp}{\textsc{{3ExpTime}}\xspace}

\newcommand{\steps}{\mathit{steps}}
\newcommand{\sh}{\mathit{sh}}

\newcommand{\app}{\mathit{app}}

\newcommand{\ini}{0}
\newcommand{\fin}{1}

\newcommand{\dirs}{\mathit{Dirs}}
\newcommand{\comms}{\mathit{Comms}}

\newcommand{\bplus}{\mathcal{B}^{+}}

\newcommand{\vps}{\text{VPS}\xspace}

\newcommand{\traces}{\mathit{traces}}

\pgfdeclarelayer{background}
\pgfdeclarelayer{foreground}
\pgfsetlayers{background,main,foreground}

\begin{document}

\maketitle

\begin{abstract}
We present novel algorithms solving the satisfiability problem and the model checking problem for Visibly Linear Dynamic Logic (\vldl) in asymptotically optimal time via a reduction to the emptiness problem for tree automata with Büchi acceptance.
Since \vldl allows for the specification of important properties of recursive systems, this reduction enables the efficient analysis of such systems.

Furthermore, as the problem of tree automata emptiness is well-studied, this reduction enables leveraging the mature algorithms and tools for that problem in order to solve the satisfiability problem and the model checking problem for \vldl.
\end{abstract}

\section{Introduction}

Visibly Linear Dynamic Logic (\vldl)~\cite{WeinertZimmermann16a} is an expressive formalism for specifying properties of recursive systems that allows for an intuitive and modular specification of an important subclass of context-free properties.
Although there exist tight bounds on the asymptotical complexity of the satisfiability- and the model checking problem for \vldl properties~\cite{WeinertZimmermann16a}, the upper bounds for both problems are witnessed by algorithms that rely on an intricate reduction of the problems to the emptiness problem for visibly pushdown automata~\cite{AlurMadhusudan04}, for which tool support is lacking.

We present novel reductions of the problems of \vldl satisfiability and \vldl model checking to the emptiness problem for tree automata~\cite{Thomas90}, yielding algorithms for both problems running in asymptotically optimal time.
Moreover, as the emptiness problem for tree automata reduces to the problem of solving two-player games with perfect information~\cite{Niessner01}, which is of great importance in the fields of program verification and program synthesis and enjoys mature tool support, the algorithms yielded by our reductions allow us to leverage this tool support for solving the problems of \vldl satisfiability and \vldl model checking.

\vldl is an extension of Linear Temporal Logic (\ltl)~ \cite{Pnueli77}, the de-facto standard for the specification of properties of non-recursive systems.
Although popular, it is lacking in expressivity, as it cannot even express all~$\omega$-regular properties.
The logic \vldl addresses this shortcoming by guarding the temporal operators of~\ltl with visibly pushdown automata (\vpas)~\cite{AlurMadhusudan04}.
A \vpa is a pushdown automaton that operates over a predefined partition of an alphabet into calls, returns, and local actions, and has to push (pop) a symbol onto (off) its stack whenever it reads a call (return).
Upon processing local actions, the automaton must not touch the stack.

Due to these restrictions, a \vldl formula can be compiled into an equivalent \vpa over infinite words of exponential size~\cite{WeinertZimmermann16a}.
As a first step in this construction, the \vldl formula is translated into a \oneaja~\cite{bozzelli07}, an automaton without stack that is able to jump from a call to its matching return.
This automaton can then be transformed into a \vpa of exponential size~\cite{bozzelli07}.
Since each visibly pushdown automaton is a classical pushdown automaton, the emptiness problem for \vpas is decidable in polynomial time~\cite{AlurMadhusudan04}.
The translation from \oneajas to \vpas, however, is quite involved, as it works for a far more complex model than is needed for the translation of \vldl formulas into \oneajas, thus hampering efforts towards an implementation of the translation from \vldl to \vpas.
This effort is further encumbered by the scant availability of emptiness checkers and model checkers for pushdown systems.

In this work, we introduce novel algorithms solving both the emptiness problem and the model checking problem for \vldl formulas in asymptotically optimal time using a translation of \vldl formulas to nondeterministic tree automata with Büchi acceptance.
The technical core of this translation is formed by an encoding of words over visibly pushdown alphabets into trees that is adapted from the encoding of such words given by Alur and Madhusudan~\cite{AlurMadhusudan04}, as well as by a translation of the \oneajas constructed from \vldl formulas into tree automata using an adaptation of the breakpoint-construction by Miyano and Hayashi~\cite{MiyanoHayashi84} in order to remove alternation and obtain a nondeterministic automaton.
Satisfiability of a \vldl formula is then checked by checking the resulting tree automaton for emptiness.
For model checking a visibly pushdown system against a \vldl specification, we translate the negation of the specification as well as the visibly pushdown system into tree automata, which we intersect and check for emptiness.

Thus, we reduce both the satisfiability- and the model checking problem for \vldl to the emptiness problem for nondeterministic tree automata with Büchi acceptance.
Hence, we reduce the complex formalism of \vldl to the simple model of nondeterministic tree automata.
Moreover, since the problem of tree automata emptiness reduces to that of solving Büchi games, which is solvable efficiently~\cite{ChatterjeeHenzingerPiterman06,ChatterjeeHenzinger12} and enjoys mature tool support~\cite{FriedmannLange09,Keiren09}, our novel reductions enable an efficient implementation of satisfiability checkers and model checkers for \vldl.

\paragraph{Related Work}

There exist a number of logics other than~\vldl that capture the class of visibly pushdown languages, most prominently \vltl~\cite{BozzelliCesar17}, a fixed-point logic~\cite{bozzelli07} and monadic second order logic augmented with a binary matching predicate (MSO$_\mu$)~\cite{AlurMadhusudan04}.
We focus here on the logic~\vldl, as it most naturally extends the concepts used by \ltl~\cite{Pnueli77}, the de-facto standard for the specification of non-recursive properties.

Moreover, there exist tools for model checking recursive problems, e.g., \textsc{Bebop}~\cite{BallRajamani00,BallRajamani01} and \textsc{Moped}~\cite{Schwoon02,SuwimonteerabuthSchwoonEsparza05}.
These tools are, however, no longer under active development, and have, to the best of our knowledge, not found widespread adoption.
In combination with the intricate translation of alternating automata into \vpas, this motivates the development of the novel translation of \vldl formulas into tree automata presented in this work.

A number of problems have been reduced to the emptiness problem for tree automata, as they are a natural model for capturing the branching-time behavior of systems~\cite{Niessner01}.
Moreover, the theory of tree automata is well-studied, with its most famous result being equivalence of tree automata and monadic second order logic of two successors~\cite{Thomas97,Zielonka98}.
Finally, the emptiness problem for tree automata with Büchi acceptance reduces to the problem of solving two-player Büchi games with perfect information~\cite{FijalkowPinchinatSerre13}.
Such games can be solved efficiently~\cite{ChatterjeeHenzinger12} and, since Büchi games are a special case of the ubiquitous parity games, there exists mature tool support for solving them~\cite{FriedmannLange09,Keiren09}.

\paragraph{Our Contributions}

Firstly, in Section~\ref{sec:stack-trees} we adapt the tree-encoding of words over visibly pushdown alphabets first introduced by Alur and Madhusudan~\cite{AlurMadhusudan04} and show that the resulting trees are recognizable by a tree automaton with Büchi acceptance condition in Theorem~\ref{thm:stack-tree-regularity}.

Secondly, in Section~\ref{sec:vldl-satisfiability}, we show how to construct tree automata recognizing the encodings of all words satisfying a given \vldl formula in Theorem~\ref{thm:vldl-to-tree-automata}.
Moreover, we show that the resulting automaton is of exponential size measured in the size of the original formula and we show that this translation yields an asymptotically optimal algorithm for satisfiability checking of \vldl formulas.

Finally, in Section~\ref{sec:vldl-model-checking} we provide a translation of visibly pushdown systems into tree automata recognizing the encodings of all traces of the system.
When combined with the previously presented translation of \vldl formulas into tree automatas, we obtain an asymptotically optimal algorithm for model checking visibly pushdown systems against \vldl specifications.
This result is given in Theorem~\ref{thm:vldl-model-checking}.

\section{Preliminaries}
\label{sec:preliminaries}

In this section we introduce the basic notions used in the remainder of this work, namely (nondeterministic) visibly pushdown automata and related concepts.

\subsection{Visibly Pushdown Languages}

A pushdown alphabet $\vpdalphabet = (\calls, \returns, \locals)$ is a finite set $\Sigma$ that is partitioned into calls $\calls$, returns $\returns$ and local actions $\locals$.
We write $w = w_0\cdots w_n$ and $\alpha = \alpha_0\alpha_1\alpha_2\cdots$ for finite and infinite words, respectively, and define the stack height reached by any automaton after reading $w$ by $\sh(w)$ inductively as $\sh(\epsilon) = 0$, $\sh(wc) = \sh(w) + 1$ for $c \in \calls$, $\sh(wr) = \max \set{0, \sh(w) - 1}$ for $r \in \returns$, and $\sh(wl) = \sh(w)$ for $l \in \locals$.
Let~$\alpha$ be a finite or infinite word.
We say that a call $\alpha_k \in \calls$ at some position $k$ of~$\alpha$ is matched if there exists a $k' > k$ such that $\alpha_{k'} \in \returns$ and $\sh(\alpha_0\cdots \alpha_{k-1}) = \sh(\alpha_0\cdots \alpha_{k'})$ and call the return at the smallest such position $k'$ the matching return of $c$.
Otherwise we call~$c$ an unmatched call.
If~$\alpha_k$ is a matched call with~$\alpha_{k'}$ as its matching return, we call the infix~$\alpha_{k+1}\cdots\alpha_{k'-1}$ of~$\alpha$ the nested infix of position~$k$.
A word is well-matched if it does not contain a return that is not a matching return.

A visibly pushdown system (\vps) $\vpsys = (Q, \vpdalphabet, \Gamma, \Delta, q_\initmark)$ consists of a finite set~$Q$ of states, a pushdown alphabet $\vpdalphabet$, a stack alphabet $\Gamma$, which contains a stack-bottom marker $\bot$, a transition relation $ \Delta \subseteq (Q \times \calls \times Q \times (\Gamma \setminus \set{\bot})) \cup (Q \times \returns \times \Gamma \times Q) \cup (Q \times \locals \times Q) $, and an initial state~$q_\initmark \in Q$.
A configuration $(q, \gamma)$ of $\vpsys$ is a pair of a state $q \in Q$ and a stack content $\gamma \in \Gamma_c = (\Gamma \setminus \set{\bot})^* \cdot \bot$.
The \vps $\vpsys$ induces the configuration graph $\graph_\vpsys = (Q \times \Gamma_c, E)$ with $E \subseteq ((Q \times \Gamma_c) \times \Sigma \times (Q \times \Gamma_c))$ and $((q, \gamma), a, (q', \gamma')) \in E$  if and only if either
\begin{itemize}[nolistsep]
\item $a \in \calls$, $(q, a, q', A) \in \Delta$, and $A\gamma = \gamma'$,
\item $a \in \returns$, $(q, a, \bot, q') \in \Delta$, and $\gamma = \gamma' = \bot$,
\item $a \in \returns$, $(q, a, A, q') \in \Delta$, $A \neq \bot$, and $\gamma = A\gamma'$, or
\item $a \in \locals$, $(q, a, q') \in \Delta$, and $\gamma = \gamma'$.
\end{itemize}
For an edge $e = ((q, \gamma), a, (q', \gamma'))$, we call~$a$ the label of~$e$.
A run $\pi = (q_0, \gamma_0)\cdots(q_n, \gamma_n)$ of $\vpsys$ on $w = w_0\cdots w_{n-1} \in \Sigma^*$ is a sequence of configurations where~$q_0 = q_\initmark$ and where $((q_i, \gamma_i), w_i, (q_{i+1}, \gamma_{i+1})) \in E$ in $\graph_\vpsys$ for all $i \in [0;n-1]$.
Infinite runs of~$\vpsys$ on infinite words are defined similarly.
We define~$\traces(\vpsys)$ as the set of all infinite words~$\alpha$ for which there exists a run of~$\vpsys$ on~$\alpha$.
Moreover, we define~$\card{\vpsys} = \card{Q}$.

\subsection{Visibly Linear Dynamic Logic}

Let $P$ be a finite set of atomic propositions and let $\vpdalphabet = (\calls, \returns, \locals)$ be a partition of $\Sigma = 2^P$.
The syntax of \vldl~\cite{WeinertZimmermann16a} is defined by the grammar
$
\varphi \coloneq p \mid \neg \varphi \mid \varphi \wedge \varphi \mid \varphi \vee \varphi
  \mid \ddiamond{\aut} \varphi 
  \mid \bbox{\aut} \varphi,
$
where $p \in P$ and $\aut$ ranges over testing visibly pushdown automata (\tnvpa) over the fixed alphabet $\vpdalphabet$.
A \tnvpa $\aut = (Q, \vpdalphabet, \Gamma, \Delta, q_\initmark, Q_F, t)$ consists of a \vps~$\vpsys = (Q, \vpdalphabet, \Gamma, \Delta, q_\initmark)$, a set of final states $Q_F \subseteq Q$, and a function $t$ mapping states to \vldl formulas over $\vpdalphabet$.
We define $\card{\varphi}$ as the sum of $\card{\cl(\varphi)}$ and the sum of the numbers of states of the automata contained in $\varphi$, where $\cl(\varphi)$ is the set of all subformulas of $\varphi$, including those contained as tests in automata and their subformulas.
We require this relation subformula-relation to be noncircular.
A run of~$\aut$ on a finite word~$w$ is a run of the underlying \vps~$\vpsys$ on~$w$.
Such a run is accepting if its final state is in~$Q_F$.

Let~$\varphi$ be a \vldl formula, let $\alpha = \alpha_0\alpha_1\alpha_2\cdots \in \Sigma^\omega$ and let $k \in \nats$ be a position in~$\alpha$.
We define the semantics of~$\varphi$ in the straightforward way for atomic propositions and Boolean connectives.
Furthermore, we define
\begin{itemize}[nolistsep]
\item $(\alpha, k) \models \ddiamond{\aut}\varphi$ if there exists $k' \geq k$ s.t.\ $(k, k') \in \Raut_\aut(\alpha)$ and $(\alpha, k') \models \varphi$,

\item $(\alpha, k) \models \bbox{\aut}\varphi$ if for all $k' \geq k$, $(k, k') \in \Raut_\aut(\alpha)$ implies $(\alpha, k') \models \varphi$,

\end{itemize}
with
\begin{multline*}
\Raut_\aut(\alpha) \coloneq \{ (k, k') \in \nats \times \nats \mid \exists \text{ acc. run } (q_0, \sigma_0)\cdots (q_{k' - k}, \sigma_{k' - k}) \\ \text{ of } \aut \text{ on } \alpha_k\cdots\alpha_{k'-1}
	\text{and } \forall m \in [0; k' - k].\ (\alpha, k' + m) \models t(q_m) \} \enspace .
\end{multline*}
We write $\alpha \models \varphi$ as a shorthand for $(\alpha, 0) \models \varphi$ and say that $\alpha$ is a model of $\varphi$ in this case.
The language of $\varphi$ is defined as $L(\varphi) = \set{ \alpha \in \Sigma^\omega \mid \alpha \models \varphi}$.
If~$\lang(\varphi) \neq \emptyset$, we say that~$\varphi$ is satisfiable.

\subsection{Tree Automata}

Let~$\bools = \set{0,1}$ and let~$\Sigma$ be an alphabet.
A~$\Sigma$-tree~$t$ is a mapping~$t\colon \bools^* \rightarrow \Sigma$.
We call a finite word~$b \in \bools^*$ a node and an infinite word~$\beta \in \bools^\omega$ a branch.
Given a node~$b$, we call the nodes~$b0$ and~$b1$ the left- and right-hand children of~$b$.
Analogously, we call the trees rooted at the left- and right-hand children of~$b$ the left- and right-hand subtrees of~$b$, respectively.
Moreover,~$b$ is the parent of both~$b0$ and~$b1$.
We call the node at address~$\epsilon$ the root of~$t$.
We say that a branch~$\beta$ contains a node~$b$ if~$b$ is a prefix of~$\beta$.
Similarly, as each node~$b$ is associated with the unique path from the root of the tree to~$b$, we say that a node~$b'$ is on the path to~$b$ if~$b'$ is a prefix of~$b$.
If~$t(b) = a$, we say that~$b$ is labeled with~$a$.
Moreover, given a tree~$t$ and a node~$b$, we define the sub-tree~$\restr{t}{b}$ of~$t$ rooted at~$b$ by~$\restr{t}{b}(b') = t(bb')$.

A tree automaton (with Büchi acceptance)~$\treeaut = (Q, \Sigma, \Delta, q_\initmark, Q_F)$ consists of a finite set of states~$Q$, an alphabet~$\Sigma$, a transition relation~$\Delta \subseteq Q \times \Sigma \times Q \times Q$, an initial state~$q_\initmark \in Q$, and a set of accepting states~$Q_F \subseteq Q$.
A run~$r$ of~$\treeaut$ on a $\Sigma$-tree~$t$ is a $Q$-tree with $r(\epsilon) = q_\initmark$ and~$(r(b), t(b), r(b0), r(b1)) \in \Delta$ for all~$b \in \bools^*$.
A branch of~$r$ is accepting if it contains infinitely many nodes~$b$ such that~$r(b) \in F$.
A run is accepting if all of its branches are accepting, while an automaton~$\treeaut$ accepts a tree~$t$ if there exists an accepting run of~$\treeaut$ on~$t$.
The language~$L(\treeaut)$ of~$\treeaut$ is defined as the set of all trees accepted by~$\treeaut$.
A set of trees is regular if there exists a tree automaton recognizing it.
We define~$\card{\treeaut} = \card{Q}$.
Tree automata are closed under intersection via an adaptation of the construction for the intersection of automata on words.
Hence, for tree automata~$\treeaut_1$,$\treeaut_2$ there exists a tree automaton~$\treeaut$ with~$\card{\treeaut} \in \bigo(\card{\treeaut_1} \card{\treeaut_2})$ such that~$\lang(\treeaut) = \lang(\treeaut_1) \cap \lang(\treeaut_2)$.

\section{Stack Trees}
\label{sec:stack-trees}

Alur and Madhusudan showed how to encode words over some visibly pushdown alphabet as a tree by ``folding away'' the nested infixes of calls into subtrees, thus moving a call and its matching return next to each other in the resulting tree~\cite{AlurMadhusudan04}.
In this section, we slightly adapt their encoding in order to simplify our construction of tree automata later on in Section~\ref{sec:vldl-satisfiability}.
In that section, we construct for each \vldl formula~$\varphi$ a tree automaton that accepts precisely the encodings of words satisfying~$\varphi$.

For the remainder of this work, we fix some pushdown alphabet~$\vpdalphabet = (\calls, \returns, \locals)$ as a partition of some alphabet~$\Sigma$.
Let~$\alpha \in \Sigma^\omega$ be an infinite word and define $\Sigma_\bot = \Sigma \cup \set{\bot}$, where~$\bot$ is some fresh symbol.
Intuitively, every node in the resulting tree denotes either one position of~$\alpha$, or it is labeled with the special symbol~$\bot$.
For a given word~$\alpha \in \Sigma^\omega$, we define the function~$\StackTree$ mapping finite and infinite words over~$\Sigma$ to infinite~$\Sigma_\bot$-trees in Figure~\ref{fig:tree-encoding}.
At every matched call, we encode its matched infix and the suffix starting at and including its matched return in the right- and left-hand subtrees, respectively.
At an unmatched call, we encode the suffix of the word starting at the symbol succeeding the unmatched call in the right-hand subtree.
If the current letter is not a call, we encode the suffix starting at the current letter's successor in the left-hand subtree.
All vertices not encoding a symbol of~$\alpha$ are labeled with~$\bot$.

\begin{figure}
	\centering
	\begin{align*}
		\StackTree(c w r\alpha) = & \tikztree{$c$}{$\StackTree(r\alpha)$}{$\StackTree(w)$}
			\begin{minipage}{.25\textwidth}
				if $c \in \calls$ and $r$ \\ is
				matching \\ return of~$c$
			\end{minipage} &
		\StackTree(c\alpha) = & \tikztree{$c$}{$\StackTree(\epsilon)$}{$\StackTree(\alpha)$}
			\begin{minipage}{.17\textwidth}
				if $c \in \calls$ and $c$ \\
				is unmatched
			\end{minipage} \\
		\StackTree(x\alpha) = & \tikztree{$x$}{$\StackTree(\alpha)$}{$\StackTree(\epsilon)$}
			\begin{minipage}{.25\textwidth}
				if $x \in \locals \cup \returns$
			\end{minipage} &
		\StackTree(\epsilon) = & \tikztree{$\bot$}{$\StackTree(\epsilon)$}{$\StackTree(\epsilon)$}
	\end{align*}
	\caption{Definition of~$\StackTree\colon \Sigma^\omega \cup \Sigma^* \rightarrow T_{\Sigma_\bot}$.}
	\label{fig:tree-encoding}
\end{figure}

A tree~$t$ is a stack tree if~$t = \StackTree(\alpha)$ for some~$\alpha \in \Sigma^\omega$.
We define the set of all stack trees over~$\Sigma$ as~$\StackTree(\Sigma^\omega) = \set{ t(\alpha) \mid \alpha \in \Sigma^\omega}$.

\begin{theorem}
	\label{thm:stack-tree-regularity}
	The set~$\StackTree(\Sigma^\omega)$ is regular.
\end{theorem}

\begin{proof}
We first introduce some notation.
Let~$t$ be a~$\Sigma_\bot$-tree.
We say that a node~$b$ is a matched call if~$t(b) \in \calls$ and~$t(b0) \in \returns$.
Similarly,~$b0$ is a matched return if we have~$t(b) \in \calls$.
If all calls and returns in~$t$ are matched, we say that~$t$ is well-matched.
Furthermore, we call a branch~$\beta$ finite in~$t$ if it eventually only contains~$\bot$-labeled vertices.
Otherwise, we call~$\beta$ infinite in~$t$.
Finally, we call a tree finite if all of its branches are finite.

We claim that a~$\Sigma_\bot$-tree~$t$ is a stack tree if and only if $t(\epsilon) \neq \bot$, if there exists a single branch that is infinite in~$t$, and if the following properties hold true for all~$b \in \bools^*$:
\begin{enumerate}[nolistsep]
	\item \label{cond:bottom} If~$t(b) = \bot$, then~$t(b0) = t(b1) = \bot$,
	\item \label{cond:matched-call} if~$t(b) \in \calls$ and~$b$ is matched, then~$\restr{t}{b1}$ is finite and well-matched,
	\item \label{cond:unmatched-call} if~$t(b) \in \calls$ and~$b$ is unmatched, then~$t(b0) = \bot$ and~$\restr{t}{b1}$ contains no unmatched returns, and
	\item \label{cond:locals-returns} if~$t(b) \in \locals \cup \returns$, then~$t(b1) = \bot$.
\end{enumerate}

Note that each of these properties can be checked by a tree automaton.
As tree automata are closed under intersection, there also exists a single tree automaton that checks all of the above properties.

It remains to show that the conditions above indeed characterize stack trees, i.e., that a~$\Sigma_\bot$-tree is a stack tree if and only if it satisfies the conditions above.
First note that for all~$\alpha \in \Sigma^\omega$, the tree~$\StackTree(\alpha)$ clearly satisfies the above conditions.
Hence, we now show that for each tree~$t$ satisfying these conditions there exists a word~$\alpha \in \Sigma^\omega$ such that~$t = \StackTree(\alpha)$.
We construct such a word via a preorder traversal of~$t$ that visits right-hand children before left-hand ones.	

We first show how to encode finite-trees, as such trees encode nested infixes of matched calls.
Let~$t$ be a finite~$\Sigma_\bot$-tree satisfying conditions~$1$ through~$4$ and let~$\height(t)$ be the minimal~$k$ such that for all nodes~$b$ with~$\card{b} \geq k$ we have~$t(b) = \bot$ and.
We construct a word~$w \in \Sigma^*$ such that~$t = \StackTree(w)$ by induction over~$\height(t)$.
If~$\height(t) = 0$, then~$t(b) = \bot$ for all~$b \in \bools^*$ due to Condition~\ref{cond:bottom} and thus, $t = \StackTree(\epsilon)$.
If, however,~$\height(t) > 0$, then first note that~$t(\epsilon) \notin \returns$, since~$\epsilon$ would be an unmatched return in that case.
Thus, first assume~$t(\epsilon) = l \in \locals$.
Then~$t(1) = \bot$ and~$t' = \restr{t}{0}$ is a well-matched~$\Sigma_\bot$-tree with~$\height(t') < \height(t)$.
Hence, there exists a word~$w' \in \Sigma^*$ such that~$\StackTree(w') = t'$.
Thus, we pick~$w = lw'$ and obtain~$t = \StackTree(w)$.
Now assume~$t(\epsilon) = c \in \calls$.
As every call in~$t$ is matched, we obtain~$t(0) = r \in \returns$ and thus,~$\restr{t}{01}(b) = \bot$ for all $b \in \bools^*$, while there exist words~$w_1$ and~$w_{00}$ such that~$\StackTree(w_1) = \restr{t}{1}$ and~$\StackTree(w_{00}) = \restr{t}{00}$ due to the induction hypothesis and due to the second condition given above.
Hence, we pick~$w = c w_{1} r w_{00}$ and obtain~$t = \StackTree(w)$.

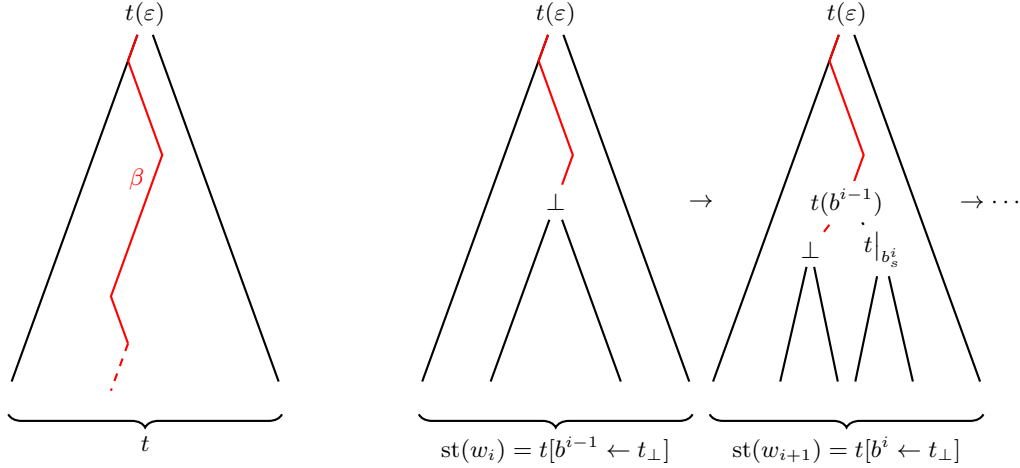
\begin{figure}
	\centering
	\begin{tikzpicture}[thick,yscale=1.25,xscale=.9]
		\begin{scope}
			\node (root) at (0,0) {$t(\epsilon)$};
			\node (bottom-left) at (-2,-4) {};
			\node (bottom-right) at (2,-4) {};
			
			\draw
				(root) -- (bottom-left)
				(root) -- (bottom-right);
				
			\draw[red]
				(root) -- (-.25,-.5) -- (0,-1) -- (.25,-1.5) -- node[anchor=east] {$\beta$} (0,-2) -- (-.25,-2.5) -- (-.5,-3) -- (-.25,-3.5);
			\draw[red,dashed] (-.25,-3.5) -- (-.5,-4);
				
			\draw[decorate,decoration={brace,amplitude=5pt,mirror}]
				(-2,-4.25) -- node[anchor=north,transform canvas={yshift=-5pt}] {$t$} (2,-4.25);
		\end{scope}

		\begin{scope}[xshift=6cm]
			\node (root) at (0,0) {$t(\epsilon)$};
			\node (bottom-left) at (-2,-4) {};
			\node (bottom-right) at (2,-4) {};
			
			\node (inner-root) at (0,-2) {$\bot$};
			\node (bottom-left-inner) at (-1,-4) {};
			\node (bottom-right-inner) at (1,-4) {};
	
			\draw
				(root) -- (bottom-left)
				(root) -- (bottom-right)
				(inner-root) -- (bottom-left-inner)
				(inner-root) -- (bottom-right-inner);
				
			\draw[red]
				(root) -- (-.25,-.5) -- (0,-1) -- (.25,-1.5) -- (inner-root);
				
			\draw[decorate,decoration={brace,amplitude=5pt,mirror}]
				(-2,-4.25) -- node[anchor=north,transform canvas={yshift=-5pt}] {$\StackTree(w_i) = t[b^{i-1} \leftarrow t_\bot] $} (2,-4.25);
		\end{scope}
		
		\node at (8.125,-2) {$\rightarrow$};
		
		\begin{scope}[xshift=10.25cm]
			\node (root) at (0,0) {$t(\epsilon)$};
			\node (bottom-left) at (-2,-4) {};
			\node (bottom-right) at (2,-4) {};
			
			\draw
				(root) -- (bottom-left)
				(root) -- (bottom-right);
			
			\draw[red]
				(root) -- (-.25,-.5) -- (0,-1) -- (.25,-1.5) -- (0,-2);
			
			\node[fill=white] (new-root) at (0,-2) {$t(b^{i-1})$};
			\node[anchor=east] (new-root-lhc) at (-.25,-2.5) {$\bot$};
			\node[xshift=.25cm] (new-root-rhc) at (.25,-2.5) {$\restr{t}{b^i_s}$};
			\node (bottom-left-inner) at (-1,-4) {};
			\node (bottom-left-middle) at (-.1,-4) {};
			\node (bottom-right-middle) at (.1,-4) {};
			\node (bottom-right-inner) at (1,-4) {};
			
			\draw
				(new-root) -- (new-root-rhc)
				(new-root-lhc) -- (bottom-left-inner)
				(new-root-lhc) -- (bottom-left-middle)
				(new-root-rhc) -- (bottom-right-middle)
				(new-root-rhc) -- (bottom-right-inner);
			
			\draw[red]
				(new-root) -- (new-root-lhc);

			\draw[decorate,decoration={brace,amplitude=5pt,mirror}]
				(-2,-4.25) -- node[anchor=north,transform canvas={yshift=-5pt}] {$\StackTree(w_{i+1}) = t[b^i \leftarrow t_\bot] $} (2,-4.25);
		\end{scope}
		
		\node at (12.375,-2) {$\rightarrow \cdots$};
		
		\node at (0,-4.5) {};
		
	\end{tikzpicture}
	\caption{Construction of the~$\StackTree(w_i)$. We have~$\beta = b_0b_1b_2\cdots$ and use the shorthands~$b^i = b_0\cdots b_i$, and~$b_s^i = b_0 \cdots b_{i-1} (1-b_i)$.}
	\label{fig:construction-illustration}
\end{figure}

Now let~$t$ be a tree that satisfies conditions~$1$ through~$4$ with~$t(\epsilon) \neq \bot$ and let~$\beta = b_0b_1b_2\cdots$ be the single infinite branch of~$t$.
As a shorthand, let~$b^i = b_0\cdots b_i$ and let~$t_a$ be the unique~$\Sigma_\bot$-tree with~$t_a(\epsilon) = a$ and~$t_a(b) = \bot$ for all~$b \in \bools^+$.
Moreover, if~$t$ and~$t'$ are~$\Sigma_\bot$-trees and~$b \in \bools^*$, we define $t[b \replacedby t']$ such that $t[b \replacedby t'](b') = t'(b'')$ if~$b' = bb''$ for some~$b'' \in \bools^*$, and $t[b \replacedby t'] = t(b)$ otherwise.
Intuitively, we replace the subtree of~$t$ anchored at~$b$ by the tree~$t'$.
We construct a series of words~$w_0,w_1,w_2,\dots \in \Sigma^*$ such that for each~$w_i$, 
\begin{enumerate}[nolistsep]
	\item $w_i$ is a strict prefix of~$w_{i+1}$
	\item $\StackTree(w_i) = t[b^{i-1} \replacedby t_\bot]$, and
	\item $\StackTree(w_i \cdot a) = t[b^{i-1} \replacedby t_a]$, where $a = t(b^{i-1})$.
\end{enumerate}
We illustrate this construction in Figure~\ref{fig:construction-illustration}.
Due to the first condition, the limit of the~$w_i$ for~$i \rightarrow \infty$ is an~$\omega$-word~$\alpha$, which, due to the second condition, satisfies~$\StackTree(\alpha) = t$.
The final condition allows us to construct the~$w_i$ inductively by collecting the labels of the nodes along the infinite path~$\beta$ of~$t$:
Upon encountering a matched return~$r$ we are able to append~$r$ to the~$w_i$ constructed so far and and ensure that the resulting~$w_{i+1}$ indeed satisfies the second condition.

Formally, we first pick~$w_0 = \epsilon$, which obviously satisfies the above requirements.
Now let~$i \in \nats$ such that~$w_i$ is defined and satisfies the above requirements.
In order to construct~$w_{i+1}$, let~$a = t(b^{i-1})$.

If~$a \in \calls$ and~$b^{i-1}$ is matched, then the subtree rooted at~$b^{i-1}1$ is finite and well-matched due to Condition~\ref{cond:matched-call}, hence there exists a word~$w \in \Sigma^*$ such that~$\StackTree(w) = \restr{t}{b^{i-1}1}$ as shown above.
Thus, it is easy to verify that~$w_{i+1} = w_i a w$ satisfies the requirements above.
In particular the third requirement is satisfied due to~$b_i = 0$, $t(b^i) \in \returns$, and due to the fact that~$w$ is well-matched, i.e., it does not contain unmatched calls.

If~$a \in \locals \cup \returns$, however, then~$w_{i+1} = w_i a$ clearly satisfies the conditions above.
In particular, the third condition is satisfied due to~$b^{i-1}$ being part of the unique infinite branch~$\beta$.
Hence,~$t(b^i) \in \returns$ can only hold true if there exists no unmatched call on the path to~$b^i$.
Thus, the third condition is indeed satisfied.
\qed
\end{proof}

From the proof of Theorem~\ref{thm:stack-tree-regularity} we furthermore obtain that for each~$\alpha \in \Sigma^\omega$, there exists exactly one branch~$\beta = b_0b_1b_2\cdots$ such that~$\StackTree(\alpha)(b_0\cdots b_{i-1}) \neq \bot$ for each~$i \in \nats$.
We call~$\beta$ the cardinal branch of~$\StackTree(\alpha)$ and we call the positions of the symbols encoded along~$\beta$ the cardinal positions of~$\alpha$.

We give an example of~$\StackTree(\alpha)$ for~$\alpha = lclrcl\cdots$ over the alphabet~$\vpdalphabet = (\calls, \returns, \locals) = (\set{c}, \set{r}, \set{l})$ in Figure~\ref{fig:encoding-example}.
The positions~$0,1,3,4$ are cardinal positions, if we assume the second~$c$ in~$\alpha$ to be unmatched.
Recall that we defined~$\sh(w)$ to be the stack height reached by any visibly pushdown automaton after processing~$w \in \Sigma^*$.
In general, Löding et al.~\cite{LoedingMadhusudanSerre04} defined the steps of a word~$\alpha = \alpha_0\alpha_1\alpha_2\cdots$ as those positions of~$\alpha$ that reach a lower bound on the stack height reached during processing the remainder of the word, i.e.,~$\steps(\alpha) = \set{k \mid \forall k' \geq k.\, \sh(\alpha_0\cdots \alpha_k) \leq \sh(\alpha_0 \cdots \alpha_{k'}}$.
A position~$k$ is a cardinal position of~$\alpha$ if and only if it is either a steps, or if~$\alpha_k$ is the matching return of some call occurring at a step.

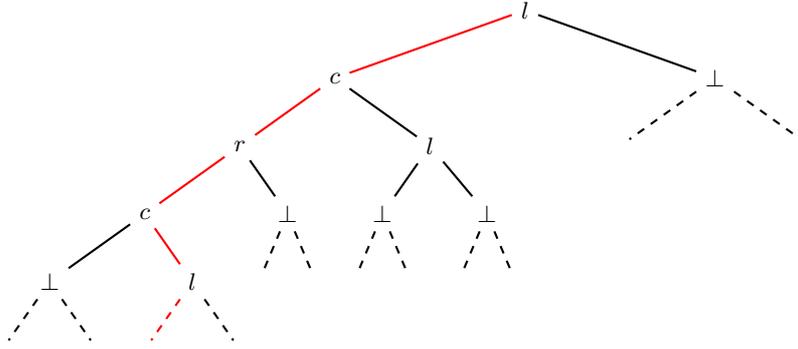
\begin{figure}
\centering
\begin{tikzpicture}[thick,xscale=2.5,yscale=.9]
	\node (root) at (0,0) {$l$};
	
	\node (0) at (-1,-1) {$c$};
	\node (1) at (1,-1) {$\bot$};
	
	\node (00) at (-1.5,-2) {$r$};
	\node (01) at (-.5,-2) {$l$};
	\node (10) at (.5,-2) {};
	\node (11) at (1.5,-2) {};
	
	\node (000) at (-2,-3) {$c$};
	\node (001) at (-1.25,-3) {$\bot$};
	\node (010) at (-.75,-3) {$\bot$};
	\node (011) at (-.2,-3) {$\bot$};
	
	\node (0000) at (-2.5,-4) {$\bot$};
	\node (0001) at (-1.75,-4) {$l$};
	\node (0010) at (-1.4,-4) {};
	\node (0011) at (-1.1,-4) {};
	
	\node (0100) at (-.9,-4) {};
	\node (0101) at (-.6,-4) {};
	\node (0110) at (-.35,-4) {};
	\node (0111) at (-.05,-4) {};
	
	\node (00000) at (-2.75,-5) {};
	\node (00001) at (-2.25,-5) {};
	\node (00010) at (-2,-5) {};
	\node (00011) at (-1.5,-5) {};
	
	\path[draw]
		(root) edge[red] (0) edge (1)
		(0) edge[red] (00) edge (01)
		(1) edge[dashed] (10) edge[dashed] (11)
		(00) edge[red] (000) edge (001)
		(01) edge (010) edge (011)
		(000) edge (0000) edge[red] (0001)
		(001) edge[dashed] (0010) edge[dashed] (0011)
		(010) edge[dashed] (0100) edge[dashed] (0101)
		(011) edge[dashed] (0110) edge[dashed] (0111)
		(0000) edge[dashed] (00000) edge[dashed] (00001)
		(0001) edge[red,dashed] (00010) edge[dashed] (00011);
\end{tikzpicture}
\caption{Encoding of~$\alpha = lclrcl\cdots$, where the second occurrence of~$c$ is an unmatched call.
The cardinal branch of~$\StackTree(\alpha)$ is marked in red.}
\label{fig:encoding-example}
\end{figure}

\section{Reducing \vldl Satisfiability to Tree Automata Emptiness}
\label{sec:vldl-satisfiability}

In this section we reduce the problem of \vldl satisfiability to the emptiness problem for tree automata.
The former problem is formulated as follows: ``Given some \vldl formula~$\varphi$, is~$\varphi$ satisfiable?''
We formalize the reduction of this problem to the emptiness problem for tree automata as follows:

\begin{theorem}
\label{thm:vldl-to-tree-automata}
For every \vldl formula~$\varphi$ there exists an effectively constructible tree automaton~$\treeaut$ such that~$\lang(\treeaut) = \StackTree(\lang(\varphi))$ with~$\card{\treeaut} \in \bigo(2^\card{\varphi})$.
\end{theorem}

Due to Theorem~\ref{thm:vldl-to-tree-automata}, we obtain an algorithm that checks \vldl formulas for satisfiability by first transforming a given formula~$\varphi$ into the tree automata~$\treeaut$ recognizing~$\StackTree(\lang(\varphi))$ and subsequently checking~$\lang(\treeaut)$ for emptiness.
Since tree automata can be checked for emptiness in polynomial time~\cite{KupfermanVardi98,ChatterjeeHenzinger12}, this algorithm runs in exponential time in~$\card{\varphi}$.
As the problem of deciding \vldl satisfiability is \exptime-hard~\cite{WeinertZimmermann16a}, the obtained algorithm is asymptotically optimal.

We split the proof of Theorem~\ref{thm:vldl-to-tree-automata} into two parts:
First, we transform a given \vldl formula into an equivalent so-called \oneaja~\cite{bozzelli07} of polynomial size.
A \oneaja is an alternating finite-state automaton on words that is able to \myquot{jump} from calls to their matching return, skipping the nested infix.
We describe this construction in the proof of Lemma~\ref{lem:vldl-to-oneaja}.
In a second step, we transform the obtained \oneaja into a tree automaton of exponential size that recognizes the stack trees of words recognized by the \oneaja.
We define this construction in Lemma~\ref{lem:oneaja-to-tree-automata}.

Let us first define the above mentioned \oneaja~\cite{bozzelli07}.
First, let~$\dirs = \set{\direct, \jump}$.
In general, we use~$\direction$ to indicate an arbitrary member of~$\dirs$.
Moreover, for a finite set~$Q$ and $\direction \in \dirs$, let~$\comms_\direction(Q) = \set{\direction} \times Q \times Q$, let~$\comms(Q) = \comms_\direct(Q) \cup \comms_\jump(Q)$, and let $\bplus( \comms(Q) )$ be the set of positive Boolean formulas over $\comms(Q)$.
Note that $\bplus( \comms(Q) )$ does not include the shorthands $\mathit{true}$ nor~$\mathit{false}$.
A 1-AJA (with Büchi acceptance) $\aut = (Q, \vpdalphabet, \delta, q_\initmark, Q_F)$ consists of
	a finite set of states~$Q$,
	a visibly pushdown alphabet~$\vpdalphabet$,
	a transition function~$\delta\colon Q \times \Sigma \rightarrow \bplus( \comms(Q) )$,
	an initial state $q_\initmark \in Q$,
	and a set of accepting states $Q_F \subseteq Q$.
We define $\card{\aut} = \card{Q}$.

Intuitively, when the automaton is in state $q$ at position $i$ of the word $\alpha = \alpha_0\alpha_1\alpha_2\cdots$, it guesses a set of commands $C \subseteq \comms(Q)$ such that $C \models \delta(q, \alpha_i)$.
It then spawns one copy of itself for each command $(\direction, q_\direct, q_\jump) \in C$ and executes the command with that copy.
If $\direction = \jump$ and if $\alpha_i$ is a matched call, the copy jumps to the position of the matching return of $\alpha_i$ and transitions to state $q_\jump$.
Otherwise, i.e., if $\direction = \rightarrow$, or if~$\alpha_i$ is not a matched call, the automaton advances to position $i+1$ and transitions to state $q_\direct$.
We say that~$\aut$ takes a jumping (direct) transition in the former (latter) case.
All copies of~$\aut$ proceed in parallel.
A single copy of~$\aut$ accepts if it visits accepting states infinitely often, while the \oneaja accepts~$\alpha$ if all of its copies accept.

Formally, a run of~$\aut$ on an infinite word $\alpha = \alpha_0\alpha_1\alpha_2\cdots$ is an infinite directed acyclic graph $R = (V, E)$ with $V \subseteq \nats \times Q$, where~$v_\initmark = (0, q_\initmark) \in V$ and all~$v \in V$ are reachable from~$v_\initmark$.
We call~$v_\initmark$ the initial vertex of~$R$ and say that a vertex~$(i, q) \in V$ is on level~$i$ of~$R$.
We require that for each $(i,q) \in V$, there exists some $C \subseteq \comms(Q)$ such that $C \models \delta(\alpha_i, q)$ and such that $((i,q), (i',q')) \in E$ if and only if $(i',q') = \app(i, c)$ for some~$c \in C$.
To this end, the command-application function~$\app$ is defined as
	$\app(i, (\direction, q_\direct, q_\jump)) = (j, q_\jump)$ if~$\direction = \jump$ and~$\alpha_i$ is a matched call with $\alpha_j$ as its matching return,
	and $\app(i, (\direction, q_\direct, q_\jump)) = (i + 1, q_\direct)$ otherwise.
We say that a vertex~$(i, q)$ is accepting if~$q$ is accepting.
Furthermore, a run~$R$ is accepting if each vertex in~$R$ has at least one successor and if all infinite paths through~$R$ starting in~$v_\initmark$ contain infinitely many accepting vertices.

Note that, in contrast to the classical definition of runs of alternating automata without jumping capability, an edge in the run of a \oneaja does not characterize an advance by a single symbol.
Instead, there exists ``long'' edges that characterize the automaton ``jumping over'' a nested infix.
Thus, there may exist positions~$k$ such that a run of a \oneaja on a word does not contain any vertices of the form~$(k, q)$, since all copies of the automaton jump over position~$k$.
The cardinal positions of a word~$\alpha$, however, serve as synchronization points of a run on~$\alpha$, as no copy of the automaton is able to jump over the cardinal points.

\begin{lemma}
\label{lem:vldl-to-oneaja}
For every \vldl formula~$\varphi$ there exists an effectively constructible \oneaja~$\aut$ with~$\lang(\aut) = \lang(\varphi)$ and with~$\card{\aut} \in \bigo(p(\card{\varphi}))$ for some polynomial~$p$.
\end{lemma}

\begin{proof}
In earlier work, we constructed a \oneaja with a more complicated condition from a given \vldl formula by induction over its structure~\cite{WeinertZimmermann16a}.
This more complicated condition allowed for a complementation without a state-space-blowup in the construction of an automaton equivalent to $\varphi = \neg \varphi'$.
As we are now aiming for a \oneaja with a simpler acceptance condition, namely a Büchi-condition, we adapt this previous construction.

In order to prevent the costly complementation of \oneaja, we require~$\varphi$ to be in negation normal form (NNF), i.e., we assume that negations only occur directly preceding atomic propositions.
Should this not be the case, we can easily transform~$\varphi$ into NNF by ``pushing down'' negations along the syntax tree, using De Morgan's law and the duality $\neg\ddiamond{\aut}\varphi = \bbox{\aut}\neg\varphi$.
Note that this latter duality does not require complementation of~$\aut$, hence it is applicable in constant time.
We then construct~$\aut_\varphi$ inductively over the structure of~$\varphi$.

If~$\varphi = p$, $\varphi = \neg p$, or $\varphi = \varphi_1 \circ \varphi_2$ for $\circ \in \set{\lor, \land}$, we trivially obtain $\aut_\varphi$, with $\card{\aut_p} = \card{\aut_{\neg p}} = 2$ and $\card{\aut_{\varphi_1 \circ \varphi_2}} \in \bigo(\card{\aut_{\varphi_1}} + \card{\aut_{\varphi_2}})$ due to closure of \oneaja under these operations~\cite{bozzelli07}.
If $\varphi = \ddiamond{\aut}\varphi'$, we follow the same intuition as in the previous construction~\cite{WeinertZimmermann16a}, i.e., we construct the \oneaja~$\aut_\varphi$ such that a single copy jumps along the cardinal positions of the input-word and spawns copies at every matched call in order to verify that the jumps taken correctly summarize finite runs of~$\aut$ on the nested infix.
Additionally,~$\aut_\varphi$ spawns copies verifying that the tests annotating the states along the simulated run hold true.
Finally,~$\aut_\varphi$ nondeterministically decides to transition into~$\aut_{\varphi'}$.
The complete construction for this case can be found in the full version of our previous work~\cite{WeinertZimmermann16a}, which can be adapted to use a Büchi-condition by making none of the states simulating~$\aut$ accepting in~$\aut_\varphi$, thus forcing the simulated run to eventually transition into~$\aut_{\varphi'}$.

If $\varphi = \bbox{\aut}\varphi'$, we obtain an automaton equivalent to~$\varphi$ via a dual construction to the one described above for the case $\varphi = \ddiamond{\aut}\varphi'$.
Again, let $\aut = (Q^\aut, \vpdalphabet, \Gamma^\aut, \Delta^\aut, q_\initmark^\aut, Q_F^\aut, t^\aut)$.
By induction, we obtain the \oneajas $\aut' = (Q', \vpdalphabet, \delta', q'_\initmark, Q'_F)$ equivalent to~$\varphi'$ and, for each test $\varphi_i$ occurring in~$\aut$, let $\aut_i = (Q^i, \vpdalphabet, \delta^i, q^i_\initmark, Q_F^i)$ be the \oneaja equivalent to $\neg\varphi_i$.
Recall that we first transform~$\neg\varphi_i$ into negation normal form by ``pushing down'' the negation to only occur in front of atomic propositions.

We now construct a \oneaja~$\aut_\varphi$ equivalent to~$\varphi$.
This construction is dual to our previous one for the case~$\varphi = \ddiamond{\aut}\varphi'$~\cite{WeinertZimmermann16a}.
Intuitively, we simulate all runs of~$\aut$ on the input word by spawning copies of the main automaton that jump along the cardinal positions of the processed input word.
Every time a call~$c$ is processed, we have to consider both cases of it being matched or unmatched.
If~$c$ is matched, for each state~$q \in Q^\aut$, we nondeterministically guess whether or not the automaton~$\aut$ can be in state~$q$ at the next step, or whether all run infixes starting in the current state lead to some state other than~$q$.
If~$c$ is unmatched, however, we treat~$c$ similarly to a local action and additionally denote that no unmatched return may be read anymore, since doing so would contradict~$c$ being unmatched.
Note that, since \oneaja process the matching return of a call after a jump instead of processing the symbol following it, we introduce waiting states that are used to delay execution for a single step.

\newcommand{\wait}{\mathit{wait}}

Formally, we define the set of states $ Q = Q_\varphi \cup \set{\top} \cup Q' \cup \bigcup\nolimits_{{\varphi_i \in \range(t^\aut)}}Q^i$,
where $Q_\varphi = \set{q, q_\wait \mid q \in (Q^\aut \times \{\ini, \fin\}) \cup (Q^\aut \times Q^\aut \times \Gamma)}$ and where the state $\top$ is used as an accepting sink.
The states from $Q^\aut \times \set{\ini, \fin}$ are used to simulate the original automaton before~$(Q^\aut \times \set{\ini})$ and after~$(Q^\aut \times \set{\fin})$ processing at least one unmatched call.

For the sake of readability, we define the transition function for the different components of the automaton separately.
Moreover, we write $(\rightarrow, q)$ and $(\rightarrow_a, q)$ as shorthands for $(\rightarrow, q, \top)$ and $(\rightarrow_a, \top, q)$.
As~$\top$ is used as a sink, we clearly have $\delta_\mathit{sink}(\top, a) = (\rightarrow, \top)$ for all $a \in \Sigma$.
Furthermore, we define~$\delta_\wait(q_\wait, a) = q$ for all~$a \in \Sigma$.

When encountering a final state of $\aut$, we model acceptance of the prefix processed so far by spawning a copy that moves to the initial state of~$\aut'$.
To achieve a uniform presentation, we define the auxiliary formula
	$\chi^f(q, a) = \delta'(q'_I, a)$ if $q \in Q_F^\aut$ and $\chi^f(q, a) = (\rightarrow, \acc)$ otherwise.
	
Moreover, we need notation to denote transitions into the automata $\aut_i$ implementing the negated tests of $\aut$.
Note that, in order to handle the test labeling the initial state of~$\aut$ correctly, we only enter the automaton implementing the negation of a test upon leaving the respective state.
Hence, we move to the successors of the initial state of~$\aut_i$ instead of moving to the initla state itself.
To this end, we define the auxiliary formula $\theta_q^a = \delta^i(q^i_I, a)$, where $t(q) = \varphi_i$.

Upon reading a local action, we have to spawn a copy to continue in the automaton~$\aut'$ if we are currently in a final state, as well as copies to track all possible continuations of the subsequent run.
If the test~$t^\aut(q)$ is violated, however, we verify that this is indeed the case by moving to the automaton implementing~$\neg t^\aut(q)$.
Hence we have
\[
 \delta_\mathit{main}((q, b), l) = \Big[ \chi^f(q, l) \land \bigwedge\nolimits_{{(q, l, q') \in \Delta}} (\rightarrow, (q', b)) \Big] \lor \theta_q^l \text{ for } l \in \locals, b \in \set{\ini, \fin}
\]

Upon reading a call, the constructed \oneaja has to consider both the case that the call is matched as well as that it is unmatched.
In the former case, for all transitions $(q, c, q', A) \in \Delta$ and all states $q'' \in Q$, the automaton either spawns a copy that verifies that it is impossible to go from $q'$ to $q''$ by popping $A$ off the stack in the final transition, or it continues at the matching return in state $q''$.
In the latter case it ignores the effects on the stack and denotes that it may not read any returns from this point onwards by setting the binary flag in its state to $\fin$.
Similarly to the previous case, we can instead verify that~$t(q)$ is violated by moving to the automaton implementing~$\neg t(q)$.
\begin{multline*}
	\delta_\mathit{main}((q, b), c)= \Big[ \chi^f(q) \land
	\bigwedge\nolimits_{{(q, c, q', A) \in \Delta, q'' \in Q}} \Big[ (\rightarrow, (q', q'', A)) \lor (\jump, (q'', b)_\wait) \Big] \land \\
	\bigwedge\nolimits_{{(q, c, q', A) \in \Delta}} (\rightarrow, (q', \fin)) \Big] \lor \theta_q^c \text{\quad  for } c \in \calls, b \in \set{\ini, \fin}
\end{multline*}

The main automaton may only handle returns as long as it has not skipped any calls.
If it encounters a return after having guessed that a call is unmatched, it moves to the accepting sink in order to be able to continue the simulation of all remaining runs.
\begin{align*}
	&\delta_\mathit{main}((q, \ini), r) = \Big[ \chi^f(q, r) \land \bigwedge\nolimits_{{(q, r, \bot, q') \in \Delta}} (\rightarrow, (q', \ini)) \Big] \lor \theta_q^r \text{\quad  for } r \in \returns \\
	&\delta_\mathit{main}((q, \fin), r) = (\rightarrow, \top) \text{\quad  for } r \in \returns
\end{align*}

The transition function $\delta_\mathit{main}$ determines the behavior of the main automaton.
It remains to define the behavior of the copies of the automaton verifying the inability of the automaton to move to some particular state upon reading a matching return.
These behave similarly to the main automaton on reading local actions and calls.
The main difference in handling calls is that these automata do not need to guess whether or not a call is matched:
Since they are only spawned on reading supposedly matched calls and terminate their run upon reading the matching return, all calls they encounter can be assumed to be matched as well.
Additionally, they never transition to the automaton $\aut'$, but merely to the automaton implementing the negation of the test of the current state upon having verified their guess.
If instead they moved, say, to the accepting sink~$\top$, there would indeed be a possibility to move to the chosen state upon popping the given stack symbol, which would contradict the nondeterministic guess made upon reading the matching call.
\begin{align*}
	\delta_\mathit{ver}((q, q', A), l) = & \Big[ \bigwedge\nolimits_{{(q, l, q'') \in \Delta}} (\rightarrow, (q'', q', A)) \Big] \lor \theta_q^l \text{\quad  if } l \in \locals \\
	\delta_\mathit{ver}((q, q', A), c) = & \Big[ \bigwedge\nolimits_{{(q, c, q'', A') \in \Delta, q''' \in Q}} (\rightarrow, (q'', q''', A')) \lor
	(\jump, (q''', q', A)_\wait) \Big] \lor \theta_q^c \text{\quad  if } c \in \calls \\
	\delta_\mathit{ver}((q, q', A), r) = &\, \theta_q^r \text{\quad  if } r \in \returns, (q, r, A, q') \in \Delta \\
	\delta_\mathit{ver}((q, q', A), r) = &\, (\rightarrow, \rej) \text{\quad  if } r \in \returns, (q, r, A, q') \not\in \Delta
\end{align*}
We then define the complete transition function $\delta$ of $\aut_\varphi$ as the union of the previously defined partial transition functions.
Since their domains are pairwise disjoint, this union is well-defined.
\[ \delta = \delta_\mathit{sink} \cup \delta_\wait \cup \delta' \cup \bigcup\nolimits_{{\varphi_i \in \range(t)}} \delta^i \cup \delta_\mathit{main} \cup \delta_\mathit{ver} \]

Finally, we make all states obtained by the translation of~$\aut$ accepting.
Thus, the simulations of all runs of~$\aut$ are accepting, which lets the complete automaton~$\aut_\varphi$ track all prefixes of the processed word that are accepted by~$\aut$.
The \oneaja 
\[ \aut_\varphi = (Q, \vpdalphabet, \delta, (q^\aut_\initmark, 0), Q_F \cup Q'_F \cup \bigcup\nolimits_{\varphi_i \in \range(t^\aut)} Q^i_F ) \]
then recognizes the language of $\varphi = \ddiamond{\aut}\varphi'$, where $Q_F = (Q^\aut \times \{\ini, \fin\}) \cup (Q^\aut \times Q^\aut \times \Gamma) \cup \set{\top}$.
\qed
\end{proof}

Having given a translation of~\vldl formulas into \oneajas, we now show how to transform a given \oneaja~$\aut$ into a tree automaton recognizing the stack trees of words recognized by~$\aut$.
To this end, consider a run~$R$ of a \oneaja~$\aut$ on some word~$\alpha \in \Sigma^\omega$, as illustrated on the left-hand side of Figure~\ref{fig:run-encoding}.
As argued above, the cardinal positions of the processed word serve as synchroniziation points in the run of~$\aut$ on~$\alpha$:
If~$i$ is a cardinal position of~$\alpha$, then there exist no positions~$j,j' \in \nats$ with $j < i < j'$ such that~$R$ contains an edge from level~$j$ to level~$j'$.
In other words, each infinite path starting in the initial vertex~$v_\initmark$ of~$R$ contains a vertex on level~$i$ for each cardinal position~$i$ of~$\alpha$.
Hence, we are able to decide whether or not~$R$ is accepting by considering finite paths of~$R$ starting and ending in levels~$i$ and~$i'$, respectively, where~$i$ and~$i'$ are cardinal positions of~$\alpha$.

More formally, we demonstrate that the breakpoint construction of Miyano and Hayashi~\cite{MiyanoHayashi84} can be adapted to \oneajas.
To this end, let~$R = (V, E)$ and let~$V_i$ be the vertices occurring in~$R$ on level~$i$, i.e.,~$V_i = \set{(i,q) \in V \mid q \in Q}$.
A breakpoint sequence over~$R$ is an infinite sequence of cardinal positions $0 = i_0 < i_1 < i_2 \cdots$ of~$\alpha$ such that all finite paths in~$R$ starting on level $i_j$ and ending on level~$i_{j+1}$ contain at least one accepting vertex.
Each cardinal position~$i_j$ in a breakpoint sequence is called a breakpoint.

\begin{lemma}
\label{lem:breakpoint-correctness}
Let~$R$ be a run of a \oneaja.
The run $R$ is accepting if and only if there exists a breakpoint sequence over~$R$.
\end{lemma}

\begin{proof}
First assume that there exists a breakpoint sequence~$0 = i_0, i_1, i_2,\dots$ over~$R$.
Then~$R$ is clearly accepting, as each infinite path~$\pi$ starting in~$v_\initmark$ is of the form~$\pi = v_0 \pi_0 v_1 \pi_1 v_2 \pi_2 \cdots$, where each $v_j \pi_j v_{j+1}$ is a path from level $i_j$ to level~$i_{j+1}$, hence $v_j \pi_j v_{j+1}$ contains at least one accepting vertex.
Thus,~$\pi$ is accepting.

For the other direction, assume that~$R$ is accepting.
We show the existence of a breakpoint sequence inductively and begin by defining~$i_0 = 0$.
Now let~$i_0, \dots, i_j$ be a finite prefix of a breakpoint sequence and assume towards a contradiction that no cardinal position~$i_{j+1}$ exists such that~$i_0,\dots,i_j,i_{j+1}$ is a prefix of a breakpoint sequence.
Then, for each cardinal position~$k$, there exists a path from some vertex on level~$i_j$ to some vertex on level~$k$ that does not contain an accepting vertex.
Hence, there also exists an infinite path starting on level~$i_j$ that does not contain an accepting vertex, which contradicts~$R$ being accepting.
Thus, there exists a cardinal position~$i_{j+1}$ such that $i_0, \dots, i_j, i_{j+1}$ is a prefix of some breakpoint sequence.
Hence, there exists a breakpoint sequence over~$R$.
\qed
\end{proof}

Given some \oneaja~$\aut$, we now construct a tree automaton that verifies that the input tree is indeed a stack tree and, if this is the case, simulates a run of~$\aut$ on the word represented by the input tree by keeping track of the set of states at each level.
Moreover, it verifies the existence of a breakpoint sequence, visiting an accepting state on the cardinal branch of the processed tree every time the corresponding symbol is at a cardinal position of the input word that can continue the prefix of the breakpoint sequence constructed so far.
In order to do so, we adapt the breakpoint construction by Miyano and Hayashi~\cite{MiyanoHayashi84}.
The key insight of this construction is that, given some breakpoint~$i$, the vertices of any level~$j > i$ can be partitioned into two sets~$A_j$ and~$N_j$.
The set~$A_{j}$ contains those states such that each finite path from some vertex on level~$i$ to some vertex on level~$j$ visits at least one accepting state, while~$N_{j}$ contains the remaining vertices on level~$j$.
We illustrate this partitioning on the left-hand side of Figure~\ref{fig:run-encoding}.
If the set~$N_{j}$ is empty, then~$j$ continues the breakpoint sequence constructed so far.
We adapt this technique in order to translate \oneaja into tree automata by keeping track of the sets~$A_i$ and~$N_i$ along the cardinal branch of the stack tree.
Upon encountering a matched call at position~$i$, the tree automaton guesses the sets~$A_j$ and~$N_j$ reached at the next cardinal position~$j$ and verifies this guess when processing the nested infix of position~$i$.

\begin{figure}
\centering
	\begin{tikzpicture}
		\begin{scope}[xscale=1.6]
			\node (0-0) at (0,0) {$(q_0, 0)$};
			
			\node (1-0) at (1,0) {$(q_1, 1)$};
			\node (1-1) at (1,1) {$(q_2, 1)$};
			
			\node (2-1) at (2,1) {$(q_3, 2)$};
			\node (2-2) at (2,2) {$(q_4, 2)$};
			
			\node (3-0) at (3,0) {$(q_5, 3)$};
			\node (3-1) at (3,1) {$(q_6, 3)$};
			\node (3-2) at (3,2) {$(q_7, 3)$};
			
			\node (4-0) at (4.5,0) {};
			\node (4-1) at (4.5,1) {};
			
			\path[draw,thick,-stealth]
				(0-0) edge (1-0) edge (1-1)
				(1-0) edge (3-0) edge (3-1)
				(1-1) edge (2-1) edge (2-2)
				(2-1) edge (3-1)
				(2-2) edge (3-2)
				(3-0) edge[dashed,shorten >= .25cm] (4-0)
				(3-1) edge[dashed,shorten >= .25cm] (4-0) edge[dashed,shorten >= .25cm] (4-1)
				(3-2) edge[dashed,shorten >= .25cm] (4-1);
				
			\begin{pgfonlayer}{background}
				\draw[thick,red!80!black,fill=red,rounded corners]
					($(3-0.south west) - (.1,.1)$) -- ($(3-2.north west) + (-.1,.1)$) --
					($(3-2.north east) + (.1,.1)$) -- ($(3-2.south east) + (.1,-.1)$) --
					($(3-2.south west) - (.05,.1)$) -- ($(3-0.north west) + (-.05,.1)$) --
					($(3-0.north east) + (.1,.1)$) -- ($(3-0.south east) + (.1,-.1)$) -- cycle;
				\node[anchor=south east,red!80!black] at ($(3-1.west) + (-.1,.2)$) {$A$};
					
				\node[draw,blue,fill=blue!20,rounded corners] at (3-1) {$(q_6,3)$};
				\node[anchor=west,blue] at (3-1.north east) {$N$};
				
				\draw[thick,red,fill=red!20,rounded corners]
					($(3-0.south west) - (.05,.05)$) -- ($(3-0.north west) + (-.05,.05)$) --
					($(3-0.north east) + (.05,.05)$) -- ($(3-0.south east) + (.05,-.05)$) -- cycle;
				\node[anchor=south,red] at ($(3-2.north) + (0,.1)$) {$A_\direct$};
					
				\draw[thick,red,fill=red!20,rounded corners]
					($(3-2.south west) - (.05,.05)$) -- ($(3-2.north west) + (-.05,.05)$) --
					($(3-2.north east) + (.05,.05)$) -- ($(3-2.south east) + (.05,-.05)$) -- cycle;
				\node[anchor=north,red] at ($(3-0.south) - (0,.1)$) {$A_\jump$};
				
				\node[draw,green,fill=green!20,rounded corners] at (3-2) {$(q_7, 3)$};
				\node[draw,green,fill=green!20,rounded corners] at (1-0) {$(q_1, 1)$};
			\end{pgfonlayer}
			
			\begin{scope}[shift={(0,-.75)}]
				\node (alpha-label) at (-.5,0) {$\alpha =$};
				\node (alpha-0) at (.5,0) {$l$};
				\node (alpha-1) at (1.5,0) {$c$};
				\node (alpha-2) at (2.5,0) {$l$};
				\node (alpha-3) at (3.5,0) {$r$};
			\end{scope}

		\end{scope}

		\begin{scope}[xshift=11.5cm,yshift=2.5cm,xscale=2]
			\node (root) at (0,0) {$l$};
			\node[anchor=west,inner sep=0] at (root.north east)
				{\scriptsize $( \set{q_0}, \emptyset)$};
			
			\node (0) at (-.5,-1) {$c$};
			\node[anchor=east,inner sep=0,align=center] at (0.north west)
				{\scriptsize $\begin{aligned}( &\set{q_2},\\ &\set{q_1}) \end{aligned}$};
			\node (1) at (.5,-1) {$\bot$};
			\node[anchor=west,inner sep=0,align=center] at (1.north east) {\scriptsize $q_\bot$};
			
			\node (00) at (-.75,-2) {$r$};
			\node[anchor=east,inner sep=0,align=center] at (00.north west)
				{\scriptsize $\begin{aligned}( &\set{q_6},\\ &\set{q_5,q_7}) \end{aligned}$};
			\node (01) at (-.25,-2) {$l$};
			\node[anchor=north west,inner sep=0,align=center] at (01.south east)
				{\scriptsize $\begin{aligned}( &\set{q_4,q_3}, \emptyset, \\ &\set{q_7}, \set{q_6}) \end{aligned}$};
				
			\node (10) at (.35,-2) {};
			\node (11) at (.65,-2) {};
			
			\node (000) at (-.875,-3) {};
			\node (001) at (-.625,-3) {};
			\node (010) at (-.375,-3) {};
			\node (011) at (-.125,-3) {};
			
			\path[draw,thick,]
				(root) edge (0) edge (1)
				(0) edge (00) edge (01)
				(1) edge[dashed] (10) edge[dashed] (11)
				(00) edge[dashed] (000) edge[dashed] (001)
				(01) edge[dashed] (010) edge[dashed] (011);
		\end{scope}
	\end{tikzpicture}	
	\caption{Encoding of a run of a \oneaja (left) into a run of a tree automaton (right). The states~$q_1$ and~$q_7$ are accepting. The positions~$0$,~$1$, and~$3$ are cardinal positions of~$\alpha$. Note $N = N_\direct = N_\jump$.}
	\label{fig:run-encoding}
\end{figure}
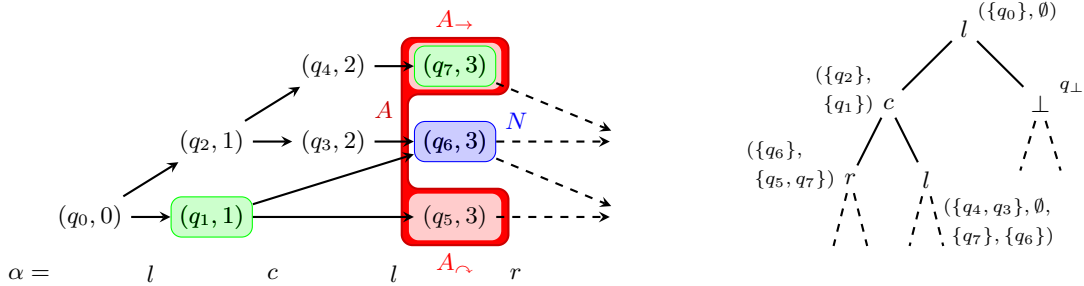

\begin{lemma}
\label{lem:oneaja-to-tree-automata}
For every \oneaja~$\aut$ there exists an effectively constructible tree automaton~$\treeaut$ with~$\lang(\treeaut) = \StackTree(\lang(\aut))$ and~$\card{\treeaut} \in \bigo(2^\card{\varphi})$.
\end{lemma}

\begin{proof}
We construct a tree automaton~$\treeaut'$ such that $\lang(\treeaut') \cap \StackTree(\Sigma^\omega) = \StackTree(\lang(\aut))$.
Recall that, due to Theorem~\ref{thm:stack-tree-regularity}, we obtain a tree automaton~$\treeaut_\Sigma$ with~$\lang(\treeaut_\Sigma) = \StackTree(\Sigma^\omega)$.
By intersecting~$\treeaut'$ with~$\treeaut_\Sigma$ we subsequently obtain~$\treeaut$ with the properties stated above.

We have explained the behavior of the automaton~$\treeaut'$ along the cardinal branch above.
It remains to take into account the effect of nested infixes on the states reached by~$\aut$ at cardinal positions.
To this end, we note that each state reached at a cardinal position is either reached by taking a jumping transition from the previous cardinal position, or by taking a direct transition from the directly preceding position, i.e., from the last position of the nested infix.
Thus, when reading a matched call at position~$i$, the automaton~$\treeaut$ guesses sets~$A_\direct, N_\direct \subseteq Q$ that are reached eventually by copies of the automaton that process the nested infix~$w$ of position~$i$.
It then assumes that these states are indeed reached by processing~$w$ and verifies that guess while processing~$\StackTree(w)$, i.e., the right-hand subtree of the matched call.
We show an example of this encoding of a run of~$\aut$ as a run of~$\treeaut'$ on the right-hand side of Figure~\ref{fig:run-encoding}.

We use two kinds of states in order to implement this idea.
States of the form~$(A, N)$, where~$A$ and~$N$ partition a nonempty subset of~$Q$ are used along the cardinal branch of the processed stack tree, implementing the breakpoint construction.
Furthermore, we use states of the form~$((A, N), (A_G, N_G))$, where~$(A,N)$ as well as~$(A_G,N_G)$ are partitions of nonempty subsets of~$Q$, to verify guesses about the effects of processing nested infixes.
Moreover, we use a sink-state~$q_\bot$ in order to process subtrees labelled exclusively with~$\bot$.

Let \oneaja~$\aut = (Q, \vpdalphabet, \delta, q_\initmark, Q_F)$.
In order to define~$\treeaut'$ concisely, we introduce some notation.
First, let~$\compl{Q_F} = Q \setminus Q_F$.
Moreover, given some set~$C \subseteq \comms(Q)$, we extract the direct- and jump-target-states $Q^C_\direction = \set{q_\direction \mid (\direction', q_\direct, q_\jump) \in C}$ for~$\direction \in \dirs$.
Furthermore, given some nonempty set~$B = \set{\varphi_0, \dots, \varphi_n}$ of Boolean formulas over~$\comms(Q)$, we write~$C \models B$ if~$C$ is a minimal model of~$B$, i.e., if ~$C = \cup_{i=0}^n C_i$ with~$C_i \models \varphi_i$.

Formally, we define the tree automaton~$\treeaut' = (Q', \Sigma, \Delta, q'_\initmark, Q'_F)$, where
$Q' = (2^Q)^2 \cup ((2^Q)^2 \times (2^Q)^2) \cup \set{q_\bot}$,
$q'_\initmark = (\set{q_\initmark}, \emptyset)$, if~$q_\initmark \in F$ and $q'_\initmark = (\emptyset, \set{q_\initmark})$ otherwise,
$Q'_F = (2^Q \times \set{\emptyset}) \cup ((2^Q)^2 \times (2^Q)^2) \cup \set{q_\bot}$,
and the transition relation~$\Delta$ defined as the smallest relation that satisfies all of the following conditions:
\begin{description}
	\item[Local Actions and Returns]
		\label{transition-function:local-and-return}
		Let~$x \in \locals \cup \returns$, let~$S$ be a nonempty subset of~$Q$, and let~$(A,N)$ be a partition of~$S$.
		Moreover, let~$C_A,C_N \subseteq \comms(Q)$ such that $C_A \models \set{\delta(q, x) \mid q \in A}$ and such that $C_N \models \set{\delta(q, x) \mid q \in N}$.
		Since~$\aut$ cannot take jumping transitions upon processing local actions or returns, we define $N' = Q^{C_N}_\direct \setminus F $ and $A' = (Q^{C_N}_\direct \cap F) \cup (Q^{C_A}_\direct \setminus N')$, thus updating the partitions~$A$ and~$N$ as described above.
		Note that the successor of a state in~$A$ may be in the subset~$N'$ if the same state is a successor of a state in~$N$.
		Moreover, it is easy to verify that~$(A',N')$ indeed are a partition of some set~$S' \subseteq Q$.
		Let~$(A_G,N_G)$ be a partition of some~$S_G \subseteq Q$.
		We require $(((A, N), (A_G, N_G)), x, ((A', N'), (A_G, N_G)), q_\bot) \in \Delta$.
		Furthermore, if~$N \neq \emptyset$, we require $((A, N), x, (A', N'), q_\bot) \in \Delta$.
	\item[Unmatched Calls]
		\label{transition-function:unmatched-call}
		Let~$c \in \calls$, let~$S$ be a nonempty subset of~$Q$, and let~$(A,N)$ be a partition of~$S$.
		Moreover, let~$C_A,C_N \subseteq \comms(Q)$ such that $C_A \models \set{\delta(q, x) \mid q \in A}$ and such that $C_N \models \set{\delta(q, x) \mid q \in N}$.
		In this case, we guess that the currently processed call is unmatched.
		Thus, similarly to the previous case,~$\aut$ can only take direct transitions.
		Define $N' = Q^{C_N}_\direct \setminus F $ and $A' = (Q^{C_N}_\direct \cap F) \cup (Q^{C_A}_\direct \setminus N')$.
		We require $((A, N), c, q_\bot, (A', N')) \in \Delta$.
		We do, however, not require a transition processing unmatched calls when verifying some guess along a nested infix, as unmatched calls cannot occur in nested infixes.
	\item[Matched Calls]
		\label{transition-function:matched-call}
		Let~$c \in \calls$ let~$S \subseteq Q$ be nonempty, and let~$(A,N)$ be a partition of~$S$.
		Moreover, let~$C^\direction_A,C^\direction_N \subseteq \comms_\direction(Q)$ for~$\direction \in \dirs$ such that $(C^\direct_A \cup C^\jump_A) \models \set{\delta(q, x) \mid q \in A}$ and such that $(C^\direct_N \cup C^\jump_N) \models \set{\delta(q, x) \mid q \in N}$.
		We follow the same idea as in the previous two cases and first define the sets of states reached by~$\aut$ at the first position of the nested infix.
		To this end, let~$N'_\direct = Q^{C^\direct_N}_\direct \setminus F$ and let $A'_\direct = (Q^{C^\direct_A}_\direct \setminus N'_\direct) \cup (Q^{C^\direct_N}_\direct \cap F)$.
		Moreover, we define the partition of states reached by~$\aut$ at the matching return of the current letter by taking a jumping transition as $N'_\jump = (Q^{C^\jump_N}_\jump \setminus F)$ and $A'_\jump = (Q^{C^\jump_A}_\jump \setminus N'_\jump) \cup (Q^{C^\jump_N}_\jump \cap F)$.
		Finally, we guess that~$\aut$ finishes processing the nested infix with the partition~$(A_G,N_G)$ of some nonempty~$S_G \subseteq Q$.
		Thus, we require $((A, N, A'_G, N'_G), c, (A'_\jump \cup A_G, N'_\jump \cup N_G, A'_G, N'_G), (A'_\direct, N'_\direct, A_G, N_G)) \in \Delta$ for arbitrary nonempty~$S'_G \subseteq Q$ where $(A'_G,N'_G)$ is a partition of~$S'_G$.
		Moreover, if~$N \neq \emptyset$, we require $((A, N), c, (A'_\jump, N'_\jump), (A'_\direct, N'_\direct, A_G, N_G)) \in \Delta$. 
	\item[Breakpoint]
		\label{transition-function:breakpoint}
		Let~$x \in \Sigma$, let~$S \subseteq Q$ be nonempty, and let~$A = S \cap F$ and~$N = S \setminus F$.
		If $((A, N), x, q'_0, q'_1) \in \Delta$, we require~$((S, \emptyset), x, q'_0, q'_1) \in \Delta$.
	\item[Verified Guess]
		\label{transition-function:nested-infix:epsilon}
		Let~$S \subseteq Q$ be nonempty and let~$(A, N)$ be a partition of~$S$.
		We require $((A, N, A, N), \bot, q_\bot, q_\bot) \in \Delta$.
	\item[Sink]
		\label{transition-function:sink}
		We require $(q_\bot, \bot, q_\bot, q_\bot) \in \Delta$.	
\end{description}

Let~$\alpha$ be some word accepted by~$\aut$ and let~$t = \StackTree(\alpha)$.
Then one can easily construct a run of~$\treeaut'$ on~$t$ from a run of~$\aut$ on~$\alpha$ as indicated in Figure~\ref{fig:run-encoding}.
In fact, if there exists an accepting run of~$\aut$ on~$\alpha$, then there also exists an accepting run of~$\treeaut'$ on~$t$, due to the implementation of the breakpoint construction and due to Lemma~\ref{lem:breakpoint-correctness}.
Conversely, if~$\treeaut'$ accepts a stack tree~$t = \StackTree(\alpha)$ for some word~$\alpha \in \Sigma^\omega$, say with the accepting run~$R'$, then it is possible to reconstruct an accepting run~$R$ of~$\aut$ on~$\alpha$ from~$R'$ via a preorder-traversal of~$R'$ that traverses the right-hand children of vertices first.
Again, due to Lemma~\ref{lem:breakpoint-correctness}, the run~$R$ is accepting if and only if~$R'$ is accepting.
Hence, by intersecting~$\treeaut'$ with the automaton~$\treeaut_\Sigma$ recognizing~$\StackTree(\Sigma^\omega)$ we obtain the automaton~$\treeaut$ recognizing~$\StackTree(\lang(\aut))$.
As~$\card{\treeaut'} \in \bigo(2^\card{\aut})$, and since~$\treeaut_\Sigma$ is of fixed size, we obtain~$\card{\treeaut} \in \bigo(2^\card{\aut})$.
\qed
\end{proof}

The proof of Theorem~\ref{thm:vldl-to-tree-automata} follows from Lemma~\ref{lem:vldl-to-oneaja} and Lemma~\ref{lem:oneaja-to-tree-automata}:
Given a \vldl formula~$\varphi$, we first construct the \oneaja~$\aut$ with~$\lang(\aut) = \lang(\varphi)$ as demonstrated in the proof of Lemma~\ref{lem:vldl-to-oneaja}.
The automaton~$\aut$ is of size polynomial in~$\card{\varphi}$.
We then construct the tree automaton~$\treeaut$ with~$\lang(\treeaut) = \StackTree(\lang(\aut))$ as shown in the proof of Lemma~\ref{lem:oneaja-to-tree-automata}.
The automaton~$\treeaut$ recognizes~$\StackTree(\lang(\aut)) = \StackTree(\lang(\varphi))$ and is of size exponential in~$\card{\aut}$, i.e., of size exponential in~$\card{\varphi}$.

\section{Reducing \vldl Model Checking to Tree Automata Emptiness}
\label{sec:vldl-model-checking}

In the previous section we have reduced the problem of \vldl satisfiability checking to the emptiness problem for tree automata.
We now consider the problem of \vldl model checking, which is formulated as follows:
``Given a \vps~$\vpsys$ and a \vldl formula~$\varphi$, does~$\traces(\vpsys) \subseteq \lang(\varphi)$ hold true?''
We now show that this problem can be reduced to the emptiness problem for tree automata similarly to the reduction of the satisfiability problem for \vldl to the same problem.

\begin{theorem}
\label{thm:vldl-model-checking}
Let~$\vpsys$ be a \vps and let~$\varphi$ be a \vldl formula.
There exists an effectively constructible tree automaton~$\treeaut$ such that~$\lang(\treeaut) = \emptyset$ if and only if~$\traces(\vpsys) \subseteq \lang(\varphi)$ with~$\card{\treeaut} \in \bigo(2^\card{\varphi} p(\card{\vpsys}))$ for some polynomial~$p$.	
\end{theorem}

\begin{proof}
Recall that~$\traces(\vpsys) \subseteq \lang(\varphi)$ if and only if~$\traces(\vpsys) \cap \lang(\neg\varphi) = \emptyset$.
Moreover, recall that we can effectively construct a tree automaton~$\treeaut_{\neg\varphi}$ such that~$\lang(\treeaut_{\neg\varphi}) = \StackTree(\lang(\neg\varphi))$ due to Theorem~\ref{thm:vldl-to-tree-automata}.
We now construct a tree automaton~$\treeaut_{\vpsys}$ recognizing~$\StackTree(\traces(\vpsys))$.
By intersecting~$\treeaut_{\vpsys}$ and~$\treeaut_{\neg\varphi}$ we subsequently obtain the tree automaton~$\treeaut$ recognizing $\traces(\vpsys) \cap \lang(\neg\varphi)$.
Hence,~$\lang(\treeaut) = \emptyset$ if and only if $\traces(\vpsys) \subseteq \lang(\varphi)$.

It remains to construct~$\treeaut_\vpsys$.
Similarly to the proof of Lemma~\ref{lem:oneaja-to-tree-automata}, we first construct~$\treeaut'_\vpsys$ such that~$\lang(\treeaut'_\vpsys) \cap \StackTree(\Sigma^\omega) = \StackTree(\traces(\vpsys))$.
By intersecting~$\treeaut'_\vpsys$ with~$\treeaut_\Sigma$ we then obtain the required~$\treeaut_\vpsys$.
The idea behind the construction of~$\treeaut'_\vpsys$ is to simulate a run of~$\vpsys$ along the cardinal branch of the tree.
This is straightforward in the case of local actions and unmatched calls or returns.
Upon encountering a matched call,~$\treeaut'_\vpsys$ guesses the state reached by~$\vpsys$ upon encountering the matched return and verifies that guess on the stack tree of the nested infix.

Let~$\vpsys = (Q, \vpdalphabet, \Gamma, \Delta, q_\initmark)$.
We define~$\treeaut'_\vpsys = (Q', q_\initmark, \Delta', Q'_F)$ with~$Q' = Q \cup (Q\times Q) \cup (Q \times \Gamma) \cup (Q \times \Gamma \times Q) \cup \set{q_\bot}$,~$Q'_F = Q'$, and
\[ \Delta' = \Delta_l \cup \Delta_{\mathit{uc}} \cup \Delta_{\mathit{ur}} \cup \Delta_{\mathit{mc}} \cup \Delta_{\mathit{mr}} \cup \Delta_{s} \enspace . \]
The individual components of~$\Delta'$ are defined as follows:
We process local actions using transitions of the form
\[ \Delta_l = \set{ (q, l, q', q_\bot), ((q, q_G), l, (q', q_G), q_\bot) \mid (q, l, q') \in \Delta, q_G \in Q} \enspace . \]
Similarly, upon encountering unmatched calls or returns, we use transitions of the form 
\[ \Delta_{\mathit{uc}} = \set{(q, c, q_\bot, q') \mid (q, c, q', A) \in \Delta} \]
and
\[ \Delta_{\mathit{ur}} = \set{(q, r, q', q_\bot) \mid (q, r, \bot, q') \in \Delta} \enspace , \]
respectively.
When encountering a matched call, we guess a state~$q_G$ reached by the automaton upon processing the matching return and verify that guess using transitions from 
\begin{multline*}
	\Delta_{\mathit{mc}} = \set{(q, c, (q_G, A), (q', q_G)) \mid (q, c, q', A) \in \Delta, q_G \in Q} \cup \\
		\set{((q, q_G), c, (q'_G, A, q_G), (q', q'_G)) \mid (q, c, q', A) \in \Delta, q_G, q'_G \in Q} \enspace .
\end{multline*}
Upon encountering a matched return, we are in some state from $(Q \times \Gamma) \cup (Q \times \Gamma \times Q)$, since a matched return only occurs directly following a matched call.
Hence, we use a transition from 
\[
\Delta_\mathit{mr} = \set{ ((q, A), r, q', q_\bot) \mid (q, r, A, q') \in \Delta } \cup
	\set{ ((q, A, q_G), r, (q', q_G), q_\bot) \mid  (q, r, A, q') \in \Delta}	
\]
in order to process that matched return.
Finally, we define $\Delta_s = \set{(q_\bot, \bot, q_\bot, q_\bot)}$ to continue the run of~$\treeaut'_\vpsys$ upon encountering the sink state~$q_\bot$.

Using the intuition given above, it can easily be verified that~$\treeaut'_\vpsys \cap \StackTree(\Sigma^\omega) = \StackTree(\traces(\vpsys))$ indeed holds true.
Thus, as previously argued, we obtain the automaton~$\treeaut_{\vpsys,\neg\varphi}$ with the properties given in the statement of this lemma.
\qed
\end{proof}

Due to Theorem~\ref{thm:vldl-model-checking}, we obtain a novel asymptotically optimal algorithm for \vldl model checking:
Given a \vps~$\vpsys$ and a \vldl formula~$\varphi$, we construct~$\treeaut$ such that~$\lang(\treeaut) = \emptyset$ if and only if $\traces(\vpsys) \subseteq \lang(\varphi)$.
The automaton~$\treeaut$ can be constructed in exponential time and is of exponential size in~$\card{\varphi}$ and of polynomial size in~$\card{\vpsys}$.
Hence, we can check~$\treeaut$ for emptiness in exponential time in~$\card{\varphi}$ and in polynomial time in~$\card{\vpsys}$.
Since the problem of \vldl model checking is \exptime-complete~\cite{WeinertZimmermann16a}, this algorithm is asymptotically optimal.

\section{Conclusion}
\label{sec:conclusion}

In this work we have presented a correspondence between infinite words over a pushdown alphabet and infinite binary trees.
Moreover, we demonstrated a construction translating \vldl formulas into tree automata that are language-equivalent with respect to the above correspondence.
This construction yields novel algorithms for satisfiability- and model checking of \vldl formulas that reduce the problem to the emptiness problem for tree automata.
Thus, this construction leverages the strong connection between visibly pushdown languages and regular tree languages that was already exhibited by Alur and Madhusudan in their seminal work on the former family of languages~\cite{AlurMadhusudan04}.
Moreover, the construction demonstrates that the well-known breakpoint construction by Miyano and Hayashi~\cite{MiyanoHayashi84}, which is routinely used to remove alternation from stack-free automata, can easily be adapted to transform alternating automata over visibly pushdown words into corresponding alternation-free automata over trees representing such words.

In future work, we plan to empirically evaluate both the algorithms presented in this work as well as those presented in earlier work~\cite{WeinertZimmermann16a}, which reduce the satisfiability- and model checking problems for \vldl to the emptiness problem for visibly pushdown automata.
Recall that our novel algorithm reduces both problems to the emptiness problem for tree automata, which in turn reduces to the well-studied problem of solving a two-player Büchi game.
The latter problem is well-studied due to its important applications, e.g., in program verification~\cite{AlurHenzingerKupferman02,Vardi08} and program synthesis~\cite{KupfermanVardi05}.
Hence, there exist efficient algorithms~\cite{ChatterjeeHenzinger12} for solving them as well as mature solvers~\cite{FriedmannLange09,Keiren09}.
Thus, we expect our novel algorithm to outperform the previous approach~\cite{WeinertZimmermann16a} to the above problems.

Moreover, in previous work we investigated the problem of solving two-player games on a visibly pushdown arena in which the winning condition is given by a \vldl formula and determined this problem to be \threeexp-complete~\cite{WeinertZimmermann16a}.
We showed membership of this problem in~\threeexp by reducing it to the problem of solving visibly pushdown games against a winning condition given by visibly pushdown automata.
Currently, we are investigating whether there exists a reduction of the former problem to that of solving games in which the winning condition is given via tree automata that yields an asymptotically optimal algorithm.

\paragraph{Acknowledgements} The author would like to thank Martin Zimmermann for multiple fruitful discussions.

\bibliographystyle{splncs03}
\bibliography{literature}

\end{document}